\documentclass[twocolumn]{article}

\usepackage{amsmath}
\usepackage{amsthm}
\usepackage{amssymb}
\usepackage{graphicx}
\usepackage{url}
\usepackage{subfigure}
\usepackage{multirow}
\usepackage{balance}

\newtheorem{lemma}{Lemma}

\newcommand{\comment}[1]{}
\newcommand{\figwidth}{0.45\textwidth}

\title{Caching Stars in the Sky: A Semantic Caching Approach to Accelerate
Skyline Queries}

\author{
\hspace*{-5mm}
\begin{tabular}{ccc}
	Arnab Bhattacharya & B. Palvali Teja & Sourav Dutta \\
	\url{arnabb@iitk.ac.in} & \url{palvali.teja@gmail.com} & \url{sodutta3@in.ibm.com} \\
	{\small Dept. of Computer Science and Engineering,} & {\small Amazon Development Limited,} & {\small IBM Research Laboratory,}\\
	{\small Indian Institute of Technology, Kanpur,} & {\small Hyderabad,} & {\small New Delhi,} \\
	{\small Kanpur, India} & {\small India} & {\small India} \\
\end{tabular}
}

\date{}

\begin{document}

\maketitle

\begin{abstract}
	Multi-criteria decision making has been made possible with the advent of
	skyline queries. However, processing such queries for high dimensional
	datasets remains a time consuming task. Real-time applications are thus
	infeasible, especially for non-indexed skyline techniques where the datasets
	arrive online.  In this paper, we propose a caching mechanism that uses the
	semantics of previous skyline queries to improve the processing time of a
	new query. In addition to exact queries, utilizing such special semantics
	allow accelerating related queries. We achieve this by generating partial
	result sets guaranteed to be in the skyline sets. We also propose an index
	structure for efficient organization of the cached queries.  Experiments on
	synthetic and real datasets show the effectiveness and scalability of our
	proposed methods.
\end{abstract}

\section{Introduction}
\label{sec:intro}

To address the problem of multi-criteria decision making and user preference
queries over attributes in relations where there is no clear preference
function, B\"orzs\"onyi et al.~\cite{kossmann} introduced skyline queries.  The
classic example of a skyline query involves choosing hotels that are good in
terms of two attributes, price and distance to beach.  The query discards hotels
that are both dearer and farther than a skyline hotel.  Formally, for every
attribute, there is a preference function that states which values dominate.

Efficient indexes are difficult to built on relations available only at run-time
or on-the-fly~\cite{4354442}.  Hence, skyline queries suffer from large
processing time and I/O bottleneck.  Caching techniques improve the situation to
some extent.  However, the use of traditional tuple and page caching techniques
do not promise significant improvement for skyline queries as user interests are
unpredictable and an inexact query with even a slight modification where
preferences are over a different subset of attributes, results in a cache miss.
For example, consider the following skyline queries:
\begin{verbatim}
select * from Airlines skyline of
       Duration min, Cost min, Services max
select * from Airlines skyline of
       Duration min, Cost min, Rating max
\end{verbatim}
The new query
\begin{verbatim}
select * from Airlines skyline of
       Duration min, Cost min
\end{verbatim}
can be answered completely from the cache if the results of the previous one are
stored and intelligent semantic caching techniques are applied. 

The special \emph{semantics} of the skyline queries allow such similar or
related queries to be processed mostly from the cache using the results of the
previous queries, without accessing the database.  Although not all skyline
queries can be handled so efficiently, the use of cache does significantly
accelerate them by producing at least partial results, which is not possible
using traditional caching mechanisms.

Our contributions in this paper are as follows:
\begin{enumerate}
	\item We introduce the concept of semantic caching for skyline queries.
	\item We categorize a new skyline query into four types according to the
		content in the cache and design efficient algorithms to process each of
		them.
	\item We design an index structure for organizing the past skyline queries
		in the cache and show how this index helps in searching the cache for
		processing the new query.
\end{enumerate}

The rest of the paper is organized as follows. Section~\ref{sec:relwork} reviews
previous research on semantic caching and skyline queries.  In
Section~\ref{sec:semcache}, a cache model is designed for reusing result sets of
previous skyline queries.   Section~\ref{sec:index} describes an index structure
to organize and access the semantic descriptions of past queries efficiently.
It also describes the cache replacement policy. In Section~\ref{sec:exp}, the
performance of skyline caching is examined through experiments.  Finally, we
summarize our work and discuss future research in Section~\ref{sec:conclusion}.

\section{Background and related work}
\label{sec:relwork}

Consider a relation $R$ with preferences specified for $k$ attributes.  A tuple
$r_i =$ $(r_{i1}, r_{i2}, \dots, r_{ik})$ \emph{dominates} another tuple $r_j =
(r_{j1}, r_{j2}, \dots, r_{jk})$ (denoted by  $r_i \succ r_j$) if for all $k$
attributes, $r_{ic}$ is preferred or equal to $r_{jc}$, and for at least one
attribute $d$, $r_{id}$ is \emph{strictly} preferred to $r_{jd}$.  The
\emph{preference functions} for each attribute are specified as part of the
skyline query.  A tuple $r$ is said to be in the \emph{skyline} set of $R$ if
there does not exist any tuple $s \in R$ that dominates $r$.

Skyline queries have been imported to databases from the maximum vector problem
or Pareto curve~\cite{kung} in computational geometry.  The first algorithm was
proposed by Kung et al.~\cite{kung}.  BNL~\cite{kossmann} uses a nested loop
approach by repeatedly reading the set of tuples.  SFS~\cite{chomicki} improves
it by sorting the data based on a monotone function.  LESS~\cite{1083622}
combines the best features of these external algorithms; however, its
performance depends on the availability of pre-sorted data.  A
\emph{divide-and-conquer} approach to partition the data so as to fit into the
main memory was proposed in~\cite{chomicki}. Using index structures, algorithms such
as NN~\cite{nn} and BBS~\cite{papadias} have been proposed. 

The idea of caching query results to optimize subsequent query processing was
first studied in~\cite{finkelstein,larson}.  Several algorithms have been
proposed in~\cite{dar,qren,dar,jonsson} that uses semantic caching efficiently
and effectively for general applications.  Also dynamic caching policies have
been studied~\cite{ssdbm}.

Several intelligent structures, e.g., SkyCube~\cite{p241} and compressed
skycubes~\cite{p491}, have been proposed to efficiently compute the varying
skyline queries based on approximate correlated user queries by using the
computational dependencies among related queries.  However, complete
construction of these structures are inefficient in real-time applications.
Further, in caching scenarios, the entire cube may not fit in the limited cache
size.  In this paper, we revisit the concept of semantic caching for skyline
queries and propose novel and intelligent algorithms along with an indexing
scheme.

\section{Capturing semantics of skyline queries}
\label{sec:semcache}

In this section, we characterize a skyline query in terms of
previous skyline queries, which help relate the new query
to those in the cache.

\subsection{Characterization of queries}

\setcounter{footnote}{0}

We assume that all the skyline queries are for a single relation.\footnote{For
different relations, separate (logical) caches can be maintained.} We also
assume the \emph{distinct value condition}~\cite{p241} which states that if no
two data points have the same values for all the dimensions, then the skyline
result for dimension set $A$ is a subset of the skyline result for dimension set
$B$ when $A \subset B$. Each query is represented as the set of attributes of
skyline preferences, which we assume is not altered for a particular dimension.
This assumption holds since the preferences of users are generally the
same.\footnote{The case where preferences may vary can be handled by considering
each (attribute, preference) pair as a separate attribute.}

Given a cache modeled as a set of queries: $C = \{S_1, S_2, \dots, S_n\}$ where
each cached query $S_j$ is again a set of attributes, a new query $Q = \{a_1,
a_2, \dots, a_q\}$ can be characterized into \emph{at least} one of the following groups:
\begin{enumerate}
	\item \textbf{Exact Query}: $Q$ is an \emph{exact} query if it matches exactly with
		a cached query, i.e., $\exists S_j, Q = S_j$, indicating the
		re-occurrence of a previous query.
	\item \textbf{Subset Query}: $Q$ is a \emph{subset} query if all its attributes are
		completely contained in a cached query, i.e., $\exists S_j, Q \subset
		S_j$.
	\item \textbf{Partial Query}: $Q$ is a \emph{partial} query if some of its
		attributes are subsets of a cached query, i.e., $\exists Q' \subset Q,
		\exists S_j, Q' \subseteq S_j$.
	\item \textbf{Novel Query}: $Q$ is a \emph{novel} query if none of its attributes
		are cached, i.e., if $\forall a_i \in Q, \forall S_j, a_i \notin S_j$.
\end{enumerate}

The hierarchy of categorization is important for query processing (details in
Section~\ref{subsec:processing}).  The \emph{most} restrictive category
determines the type of the query.  For example, if a query is both an exact and
a subset query, it is treated as an exact query and a query is categorized as a 
novel query if and only if it cannot be characterized as an exact, subset or partial query.
Table~\ref{tab:query} describes an example in detail.  When a new skyline is
queried, all the semantic segments stored in the cache are scanned to determine
the type of the new query.

\begin{table*}[t]
	\centering
	\begin{tabular}{|c||c|c|c||c|}
		\hline
		\textbf{Cache} & \multicolumn{4}{c|}{$S_1 = \{1,2,3\}$, $S_2 = \{1,2\}$, $S_3 = \{3,4\}$, $S_4 = \{5,6\}$} \\
		\hline
		\hline
		\textbf{Query} & \textbf{Exact} & \textbf{Subset} & \textbf{Partial} & \textbf{Type} \\
		\hline
		\hline
		$Q_1 = \{1,2\}$ & $S_2$ & $S_1$ & $S_1, S_2$ & Exact \\
		$Q_2 = \{2,3\}$ & - & $S_1$ & $S_1, S_2, S_3$ & Subset \\
		$Q_3 = \{4,5\}$ & - & - & $S_3, S_4$ & Partial \\
		$Q_4 = \{6,7\}$ & - & - & $S_4$ & Partial \\
		$Q_5 = \{7,8\}$ & - & - & - & Novel \\
		\hline
	\end{tabular}
	\caption{Characterization of queries.}
	\label{tab:query}
\end{table*}

Table~\ref{tab:query} describes an example in detail.  The contents of the cache
are shown in the top row.  The main rows of the table depict how each query can
be categorized into the different query types.  For example, query $Q_1$ is an
exact query because it matches with $S_2$.  It is also a subset query as its
attributes are completely contained within $S_1$.  Similarly, it can be
categorized as a partial query since some of its attributes are contained in the
cached queries $S_1$ and $S_2$.  However, it will be treated as an exact query
since that is the most restrictive category.  Query $Q_2$ is similarly
classified as a subset query even though it is also a partial query.  Query
$Q_3$ is a simple partial query.  Query $Q_4$ will also be treated as a partial
query even though some of its attributes (attribute $7$) is not cached at all.
Query $Q_5$ is a novel query as it cannot be categorized into any of the three
other types.

\subsection{Semantic segments}
\label{sec:semanticsegment}

While a cached semantic query is simply a set of attributes, certain other
descriptors are also encapsulated in a data structure called the \emph{semantic
segment} for each query.  The semantic segment for a query contains the
following fields:
\begin{itemize}
	\item \emph{Attributes and preferences}: Attributes on which
		the skyline preferences are applied.
	\item \emph{Result}: A link to a table of records that constitute the
		answer to this query.
	\item \emph{Replacement value}: It is used for cache replacement methods
		(see Section~\ref{sec:cachereplacement}).
\end{itemize}

\subsection{Query processing algorithms}
\label{subsec:processing}

Based on the type of the new query, different query processing strategies are
followed as described in this section.

\subsubsection{Exact queries}

If the query is an exact query, the result set of the cached query is directly
returned as the result set of the new query.

\subsubsection{Subset queries}

If the new query $Q$ is a subset of a cached query $S_j$, then the following
lemma shows that the result set of $Q$ is a subset of the result set of $S_j$.

\begin{lemma}
	\label{lem:subset}
	If a skyline query $Q$ is a subset of another skyline query $S$, then the
	result set of $Q$ is completely contained in the result set of $S$.
\end{lemma}
\begin{proof}
	Suppose $Q = \{a_1, a_2, \dots, a_q\}$.  Since it is a subset of $S$, $S$
	can be written as $\{a_1, a_2, \dots, a_q, s_1, s_2, \dots, s_n\}$.
	Consider a tuple $v$ which is a skyline record for $Q$.  Given the distinct 
	value condition, this implies that
	there does not exist any tuple $u \succ v$ such that $u$ dominates $v$ in
	all the attributes $\{a_1, a_2, \dots, a_q\}$.  Therefore, $u$ cannot
	dominate $v$ when more attributes $\{s_1, s_2, \dots, s_n\}$ are added.
	Thus, $v$ is a skyline record for $S$ as well.

	However, there can exist a tuple $u$ which is a skyline record for $S$ but
	not for $Q$.  Assume that $t \succ u$ in $\{a_1, a_2, \dots, a_q\}$ but $u
	\succ t$ in $\{s_1, s_2, \dots, s_n\}$.  Since $u$ is dominated in all 
	attributes of $Q$ by $t$, $u$ is not a skyline record for $Q$.
	\hfill{}
\end{proof}

The next lemma shows that to determine whether a tuple from the result set of
$S_j \supset Q$ is in the result set of $Q$, only the tuples in $S_j$ need to be
checked for dominance.

\begin{lemma}
	\label{lem:onlysubset}
	If a tuple $v$ in the result set of $S$ is not a skyline for $Q
	\subset S$, then there must exist $u \in result(S)$ such that $u \succ v$.
\end{lemma}
\begin{proof}
	Suppose $v \in result(S)$ is dominated in the attributes of $Q$ by a tuple
	$t \notin result(S)$. Since $t$ is not in the result set of $S$, there must
	exist a tuple $u \in result(S)$ that dominates $t$ in all the attributes of
	$S$ including that of $Q$.  Thus, $u \succ t$ and $t \succ v$ which together
	imply $u \succ v$, which is a contradiction.
	\hfill{}
\end{proof}

Hence, if none of the tuples in the result set of $S_j$ dominate a tuple $u$ in
all the attributes of $Q$, there cannot exist any other tuple in the relation
that can dominate $u$.  Then, $u$ will be in the result set of $Q$.  Otherwise,
it will not be.

If a new query $Q$ is a subset of many cached queries $S_i, S_j$, etc., the
processing becomes even faster.  Any tuple which is in the result set of $Q$
must be in the result set of all of $S_i, S_j$, etc.  Thus, only the tuples that
are in the \emph{intersection} of the result sets of these subset queries need
to be examined.  

While subset and exact queries can be processed from the cache itself without
accessing the database at all, the advantage cannot be retained for the other
two types of queries as explained next.

\subsubsection{Partial queries}

Suppose the new query $Q$ is partial to a cached query $S_j$.  The attributes
$Q' \subset Q$ are contained in $S_j$, and is equal to $S'_j \subseteq S_j$.
Using Lemma~\ref{lem:subset}, the skyline corresponding to the attributes $Q' =
S'_j$ is a subset of the skyline set maintained for $S_j$.  This subset is
computed and it serves as the \emph{base} set.  A special case of partial
queries allows the base set to be directly available -- when the query is a
\emph{superset} of $S_j$, i.e., $Q' = S_j$.  The \emph{entire} skyline set of
$S_j$ then serves as the base set for $Q$.

Unlike the case for subset queries, the computation of the base set does not
complete the processing.  The following lemma shows there may exist a
tuple not in the base set (i.e., the skyline set for $Q'$), but is part of the
skyline set of $Q$.

\begin{lemma}
	\label{lem:superset}
	A tuple in the skyline set of $Q$ need not be in the skyline set of its
subset $Q'$.
\end{lemma}
\begin{proof}
	Suppose $Q = \{q'_1, q'_2, \dots, q'_n, q_1, q_2, \dots, q_m\}$ and its
	subset $Q' = \{q'_1, q'_2, \dots, q'_n\}$.  Consider a tuple $v$ that is in
	the skyline set of $Q$, i.e., there is no tuple $u$ that dominates $v$ in
	all the $n + m$ attributes.  However, it may well be the case that $u \succ
	v$ in the attributes $q'_1, q'_2, \dots, q'_n$ while $v \succ u$ in the
	other attributes $q_1, q_2, \dots, q_m$.  Then, $v$ will not be a skyline
	tuple for $Q'$.
	\hfill{}
\end{proof}

Thus, the base set alone is not sufficient; it is necessary to look for tuples
that satisfy the skyline criteria from the database.  Computing the base set may
then seem as a useless exercise as scanning the database cannot be avoided anyway.
However, the base set helps in two important ways.

First, since the tuples in the base set are guaranteed to be in the skyline set
of $Q$, they can be output immediately.  For real-time applications, the
implications of this concept of \emph{incremental} results are enormous.
Without accessing the database at all, some skyline records are output; while
the user is busy examining them, the other skyline tuples can be computed and
fetched from the database.

The second important advantage is the fact that the use of a base set can speed
up most of the generic skyline algorithms, such as BNL~\cite{kossmann},
SFS~\cite{chomicki}, and LESS~\cite{1083622}.  These algorithms maintain a window of
possible skyline tuples at all times found by scanning the database in order.
Since the base set fits in the memory (as it is in the cache) and is guaranteed
to contain only skyline tuples, it can significantly improve the query
processing time by serving as the initial window.  For other non-indexed
algorithms, the base set may or may not help, but will never deteriorate the
performance.

If there are two or more queries $S_i, S_j$, etc. that are partial to $Q$, base
sets can be computed from all of them.  The \emph{union} of these sets serve as
the consolidated base set which can then be used.  Since this combined base set
is larger than any of the base sets, the advantages are more pronounced.

\subsubsection{Novel queries}

Since the novel queries contain attributes on which no previous skyline operator
has been applied, the cache does not contain any information that can be used to
expedite the processing.  Consequently, such queries are completely processed
from the database.

\subsection{Need for an index structure}

Processing a new query first involves searching all the semantic segments in the
cache to determine its type.  This is a tedious task when the number of semantic
segments is large. As the number of semantic segments is exponential in number
of dimensions, it can be very large for high dimensional datasets.

However, there is an even bigger concern when the semantic segments are not
organized.  Consider two cached queries $S_1$ and $S_2$ where $S_2 \subset S_1$.
The tuples that form the result of $S_2$ are already stored in the result of
$S_1$.  However, when the semantic segments are stored na\"ively, these tuples
are maintained twice in the cache, thereby wasting precious cache memory.  The problem
is compounded when more queries that are subsets of $S_2$ are stored.

An efficient organization of the semantic segments in the cache that can avoid
storing redundant records and can retrieve the result set by comparing with
lesser number of cached queries instead of comparing with all of them is, thus,
required.  

\section{Index structure}
\label{sec:index}

The index structure that we design is a directed acyclic graph (DAG) linking the
different semantic segments.  The semantic segment for a query $S_1$ is made a
child of the semantic segment for a query $S_2$ if $S_1 \subset S_2$.  Clearly,
a semantic segment can have multiple parents, but there cannot be any
cycle.  Note that the graph may be a forest, hence a pseudo root 
node is added that acts as the parent of all root nodes to make it connected.  
In comparison to SkyCube based structures, it does not contain the entire gamut
of the user query space and is based only on the queries previously encountered,
thereby befitting cache space requirements.  

\subsection{Modified semantic segments}
\label{sec:modifiedsemanticsegment}

To maintain this index structure, in addition to the fields described in
Section~\ref{sec:semanticsegment}, two more fields are added to each semantic
segment for efficient management of links among semantic segments:
\begin{itemize}
	\item \emph{Child pointers}
	\item \emph{Bit vectors}
\end{itemize}

The child pointers link a semantic segment to its children.  For each attribute
of the query, a bit vector is maintained.  The size of the bit vector is equal
to the number of children.  The children of a node are ordered according to
their arrival.  The $i^\text{th}$ bit in the $j^\text{th}$ bit vector is set to
$1$ if and only if the $i^\text{th}$ child contains the $j^\text{th}$ attribute.
The bit vectors help to retrieve the required children for an attribute quickly.

The size of the bit vector is not constant; rather, it grows or shrinks with the
number of children.  However, since the order of the children is fixed, there is
no ambiguity about which bit refers to which child.

\subsection{Eliminating redundancy of result sets}
\label{subsec:redundancy}

The query processing algorithms use the index structure to eliminate the redundancy 
of result sets between a cached query and its subsets. If a query has a child (i.e., a
subset), then all the skyline tuples are not stored in the result set; rather,
they are distributed between itself and the child.  For example, suppose query
$S_1$ has a child $S_2$, which is a leaf node.  The skyline tuples for $S_2$ are 
stored in its result set, i.e., $r(S_2) = s(S_2)$.  However, since these records are a 
subset of the skyline tuples for $S_1$, redundancy is removed by not
storing them again in $S_1$.  Instead, only the difference of the skyline set for $S_1$ 
with $S_2$ are stored, i.e., $r(S_1) =
s(S_1) - s(S_2)$.  The complete skyline records for $S_1$ can be retrieved by
combining the result set of $S_1$ with that of $S_2$, i.e., $s(S_1) = r(S_1)
\cup r(S_2)$.  In general, when there are multiple children, the skyline records
of all of them need to be combined to retrieve the result set for the parent.

We next explain how a semantic segment is inserted into the index of the cache.
Note that a semantic segment is inserted only when it is queried. 

\begin{figure*}[t]
\centering
\subfigure[\{1,2\}]
{
\includegraphics[width=0.05\textwidth]{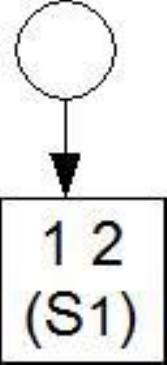}
\label{fig:i1}
}
\subfigure[\{1,2,3\}]
{
\includegraphics[width=0.05\textwidth]{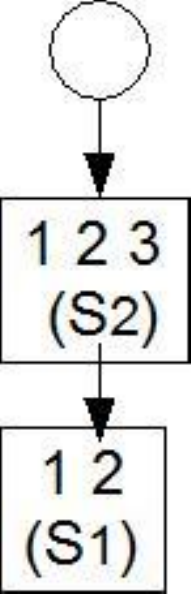}
\label{fig:i2}
}
\subfigure[\{3,4\}]
{
\includegraphics[width=0.15\textwidth]{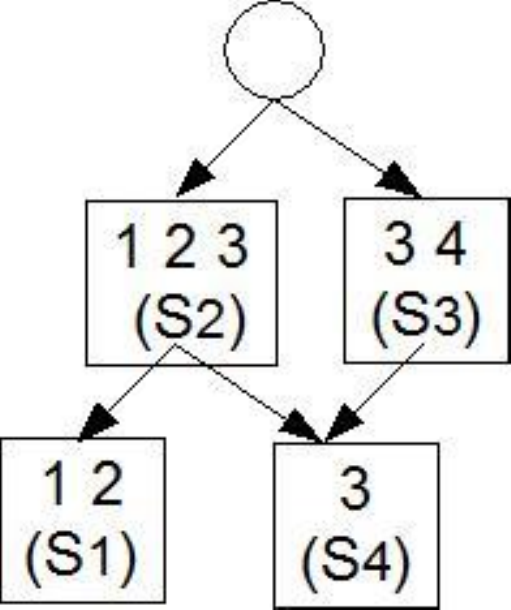}
\label{fig:i3}
}
\subfigure[\{5,6\}]
{
\includegraphics[width=0.20\textwidth]{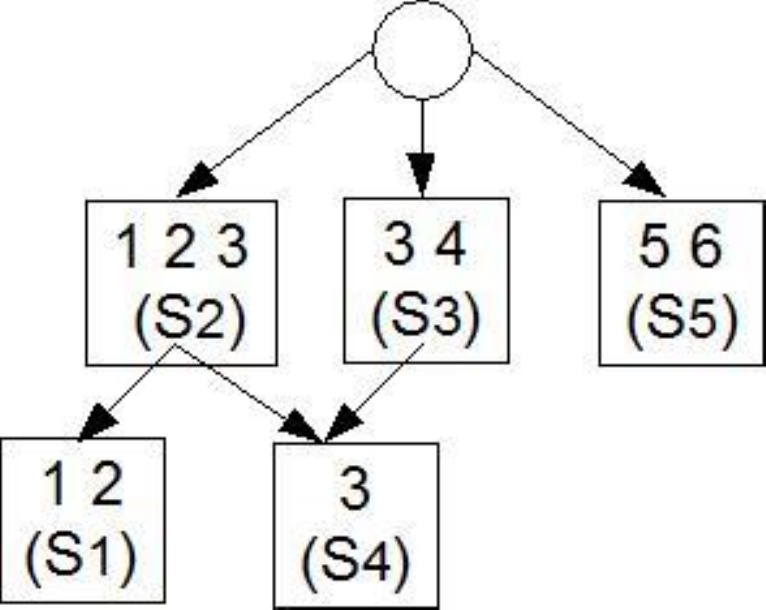}
\label{fig:i4}
}
\subfigure[\{1,2\}]
{
\includegraphics[width=0.20\textwidth]{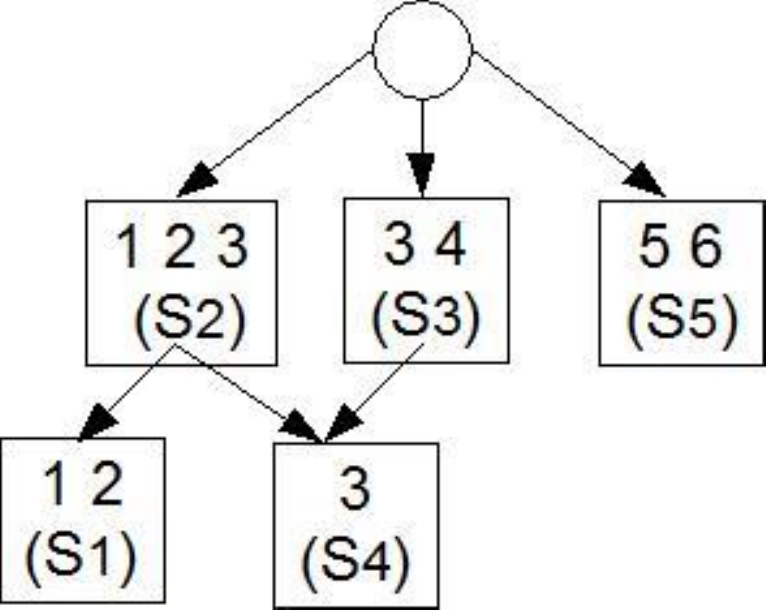}
\label{fig:q1}
}
\subfigure[\{2,3\}]
{
\includegraphics[width=0.20\textwidth]{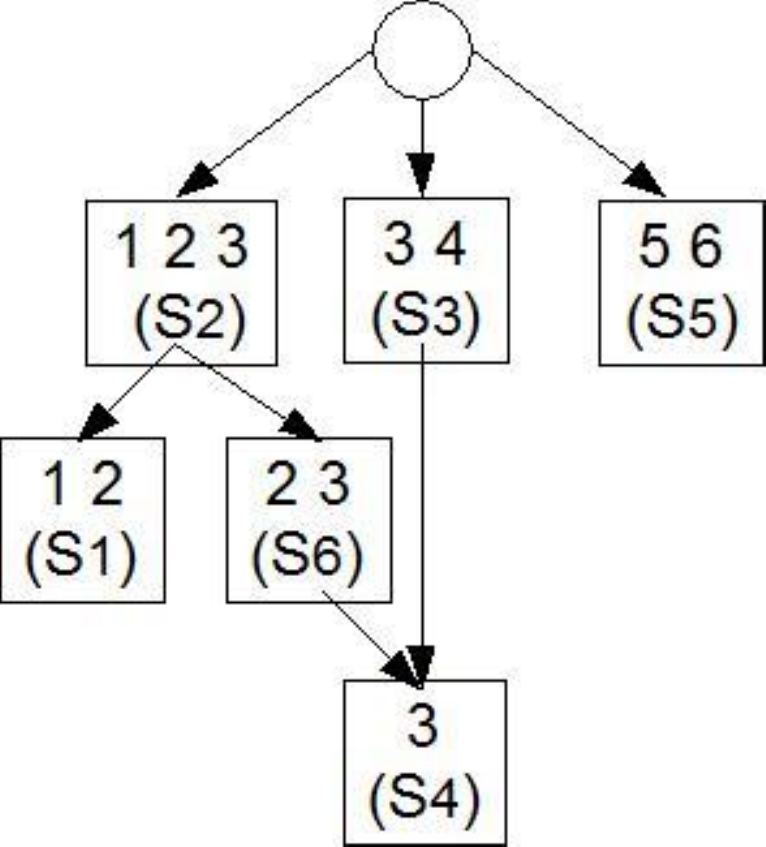}
\label{fig:q2}
}
\subfigure[\{4,5\}]
{
\includegraphics[width=0.25\textwidth]{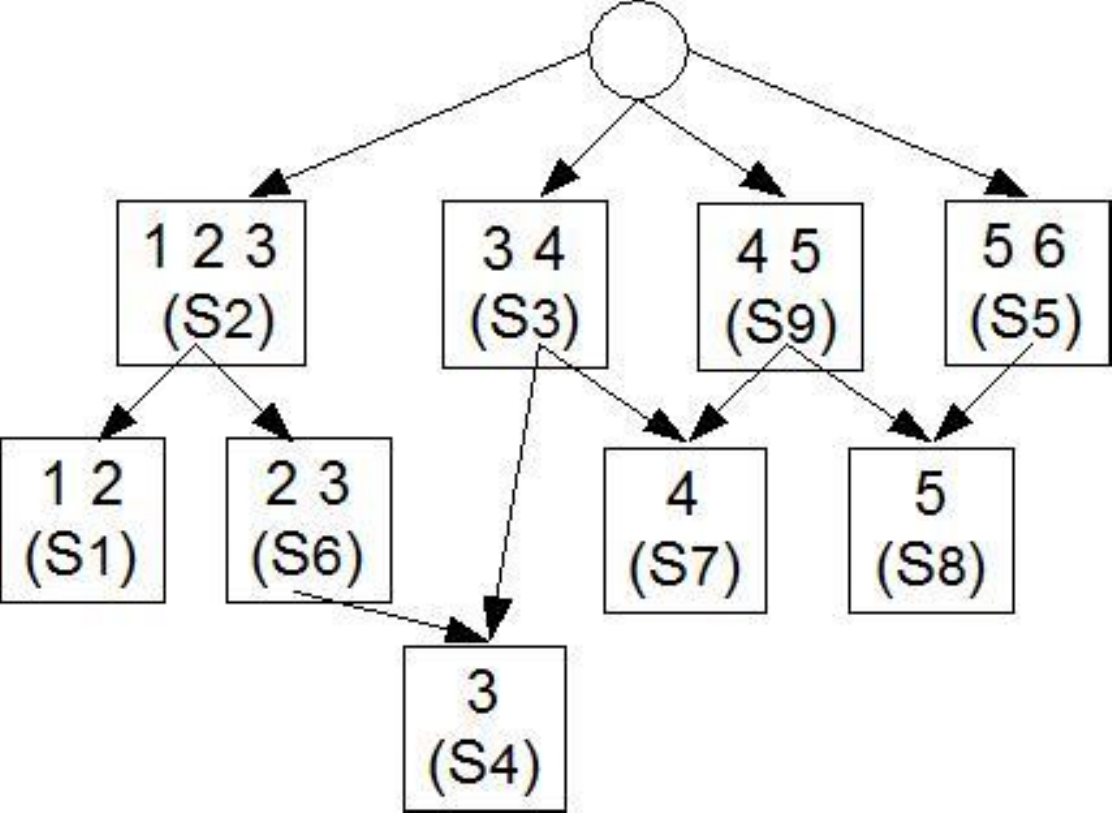}
\label{fig:q3}
}
\subfigure[$\{6,7\}$]
{
\includegraphics[width=0.30\textwidth]{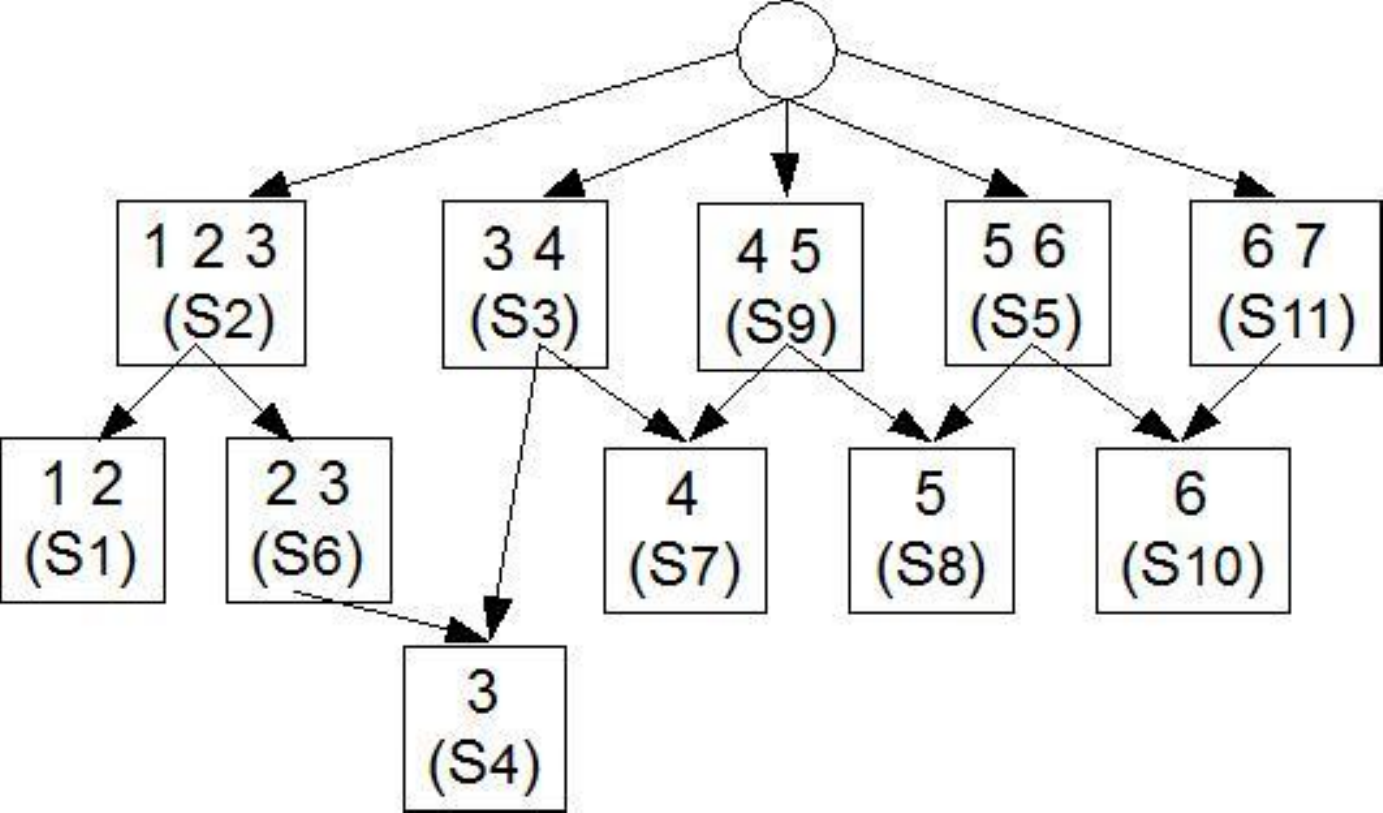}
\label{fig:q4}
}
\subfigure[\{8,9\}]
{
\includegraphics[width=0.35\textwidth]{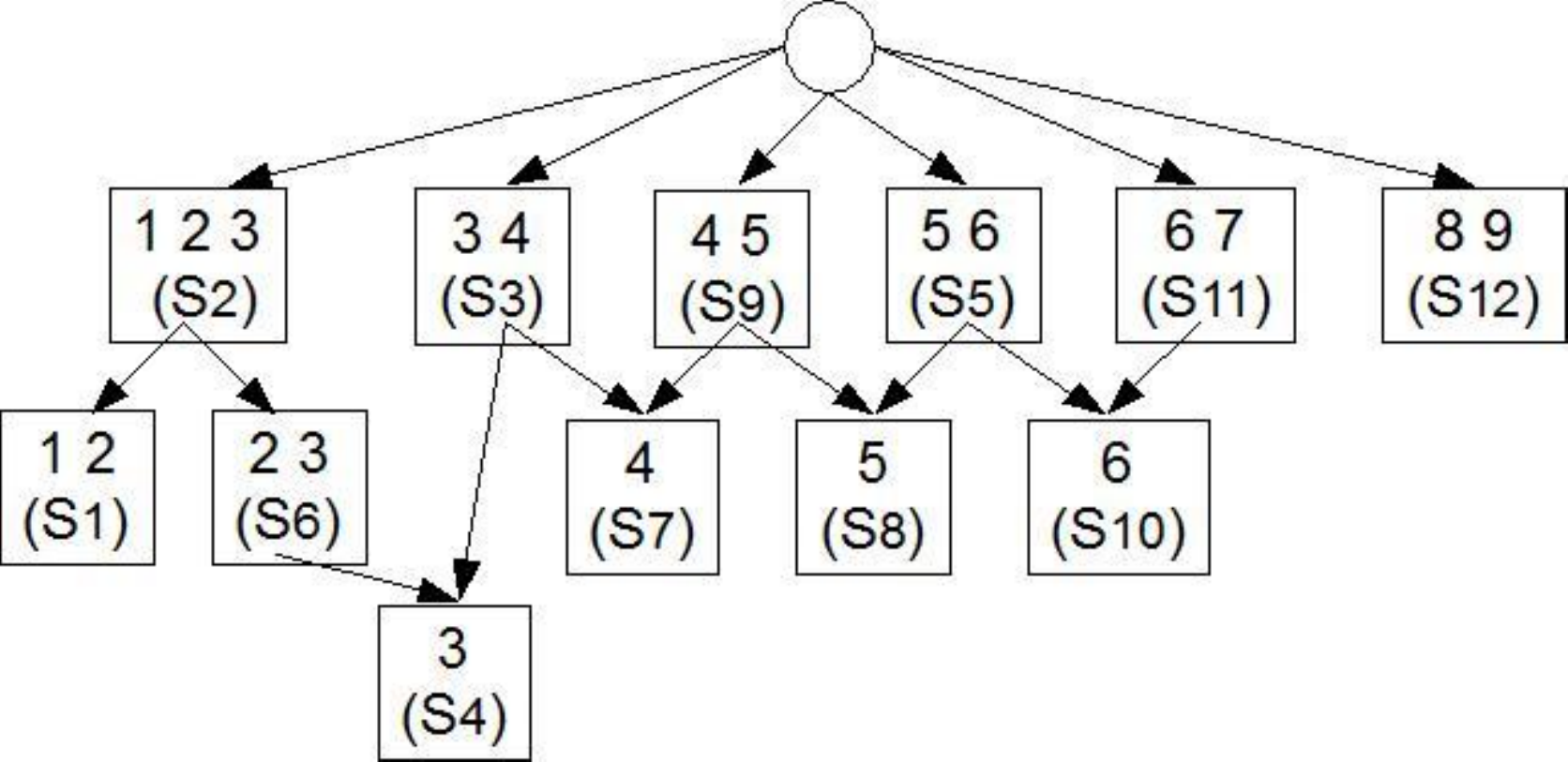}
\label{fig:q5}
}
\caption{Querying and insertion of semantic segments in the index.}
\label{fig:query}
\end{figure*}

\subsection{Query processing and insertion using index}
\label{sec:ins}

We illustrate the index search operation for query processing and subsequent
insertion using the series of query examples as shown in 
Fig.~\ref{fig:query}.  In the figures, only the attributes and the node ids
are shown for simplification.

Initially, the cache is empty and the index simply contains the pseudo root
node.  When the first query $\{1, 2\}$ arrives, it is classified as a novel
query, and is inserted as semantic segment $S_1$ (Fig.~\ref{fig:query}a).

The next query is $\{1, 2, 3\}$.  All the root nodes are searched, and it is
found out that this new query is a superset of a cached query.  Hence, it is
classified as a partial query, and the entire skyline set of $S_1$ is used as
the base set.  The new query now becomes the root and the old root its child
(Fig.~\ref{fig:query}b).

Then, query $\{3, 4\}$ arrives.  Scanning the root nodes, it is found to be
partial to $S_2$.  The base set is computed which consists of the skyline
tuples for the common attributes, i.e., $\{3\}$.  This semantic segment ($S_4$)
is a subset of both $S_2$ and the new query $S_3$ and is, therefore, maintained
as a child of both (Fig.~\ref{fig:query}c).

The next query $\{5, 6\}$ is a novel query as it does not match with
any of the root nodes.  Consequently, it is processed from the database and is
inserted as a new root node in the index (Fig.~\ref{fig:query}d).

Next is an exact query $\{1, 2\}$. The roots are scanned, and is found to be 
a subset query of the first root $S_2$.  The children of this root are then
searched to see if the categorization can be improved (as in this case).  
The skyline set of $S_1$ is returned as the answer and no
change is made to the index (Fig.~\ref{fig:query}e).

Query $\{2, 3\}$ then arrives. Being a subset of $S_2$, only the children of
$S_2$ are searched, but no exact match is found.  The skyline set of $\{2, 3\}$
is computed from that of $\{1, 2, 3\}$ and is inserted as a child of $S_2$.
Since the skyline set of $\{3\}$ is already maintained as a semantic segment
($S_4$), and it is a subset of this new query as well, the child pointers and
bit vectors are appropriately modified in $S_2$ and $S_6$ to reflect the fact
that $S_4$ now is only a descendant of $S_2$ and not a direct child
(Fig.~\ref{fig:query}f).

Queries $\{4, 5\}$, $\{6, 7\}$ and $\{8, 9\}$ are similarly handled
(Figs.~\ref{fig:query}g,~\ref{fig:query}h and~\ref{fig:query}i).

\subsection{Deletion from index}
\label{sec:del}

When the cache is full and a new query arrives, an effective replacement policy
must be chosen to select the replacement candidate.  Further, since the cache is
very dynamic, efficient update operations on the index need to be designed. 
 
The skyline set of a parent in the index is shared among itself and its
children, and the union of these sets are computed for the result.  Therefore,
if a child is deleted from the cache, for correctness, its skyline set has to
merged back with that of its parent.  Since the size of the skyline set is the
largest factor for the size of a semantic segment, deleting a child does not
produce much advantage.  Thus, for our index structure, we only delete the root
nodes and the children become the new roots if they have no parent.

\begin{table*}[t]
\centering
\begin{tabular}{|l|l|}
\hline
\multicolumn{1}{|c|}{\textbf{Parameter}} & \multicolumn{1}{c|}{\textbf{Values}} \\
\hline
\hline
Cardinality ($N$) & 1 $\times$ 10$^4$, 3 $\times$ 10$^4$, \textbf{1} $\times$ \textbf{10}$^5$, 3 $\times$ 10$^5$, 1 $\times$ 10$^6$ \\
Dimensionality ($d$) & 3, 4, 5, \textbf{6}, 7 \\
Cache size ($|C|$) & 0.1\%, 1\%, 3\%, \textbf{5\%}, 7\%, 10\% \\
Number of queries ($|Q|$) & 1, 5, 10, 25, 50, \textbf{100} \\
\hline
\end{tabular}
\caption{Experimental parameters and their default values (in bold).}
\label{tab:expparams}
\end{table*}

\subsection{Cache replacement}
\label{sec:cachereplacement}

Due to limited cache size, not all semantic segments encountered can be stored.
This is the main drawback of SkyCube-based techniques.  For efficient use of
cache, the most useful semantic segments need to be preserved and the rest
should be replaced. 

The first important parameter is the \emph{usage factor} ($\alpha$).  When the
semantic segment is first introduced into the index, its replacement factor is
set to $1$.  Every time its result set is used, the value is incremented.  The
one with a lower replacement factor should be replaced, as it is being less
used.

The second important factor is the \emph{size} ($\beta$) of the skyline set,
i.e., the number of tuples in it.  Since the available memory in the cache is a
premium asset, a semantic segment that stores a large number of tuples as its
skyline set does not allow other semantic segments to be stored.  Hence, it
should be removed. 

The third parameter that determines the usefulness of a semantic segment is
\emph{dimensionality} ($d$).  When the number of dimensions is more, there is
more chance of a new query to become a subset of it or to have more overlap in
case it is a partial query and, therefore, should not be replaced.

A \emph{replacement value} ($\delta$) for each semantic segment is computed by
combining the three, i.e., $\delta = f(\alpha, \beta, d)$.  The semantic segment
with the \emph{lowest} $\delta$ is the \emph{least useful} and should be chosen
for replacement.  The function $f$, therefore, should be monotonic with $\alpha$
and $d$ and anti-monotonic with $\beta$.  While different functions fit the
condition, the following simple function empirically produces good results:
\begin{align}
	\delta = (\alpha \times d) / \beta \nonumber
\end{align}

\section{Experimental results}
\label{sec:exp}

In this section, we evaluate the performance of the caching techniques.  The
techniques were implemented using Java on an Intel Core 2 Duo 2GHz machine with
2GB RAM in Ubuntu Linux environment.  For skyline computation, we used the
non-indexed sort-filter-skyline (SFS)~\cite{chomicki} algorithm.  We analyzed and
compared the execution times of three different skyline processing techniques:
(i)~without using cache (NC), (ii)~using cache without using the index (NI), and
(iii)~using cache with index (Index).  

\begin{figure*}[t]
\centering
\begin{tabular}{cc}
	\includegraphics[width=\figwidth]{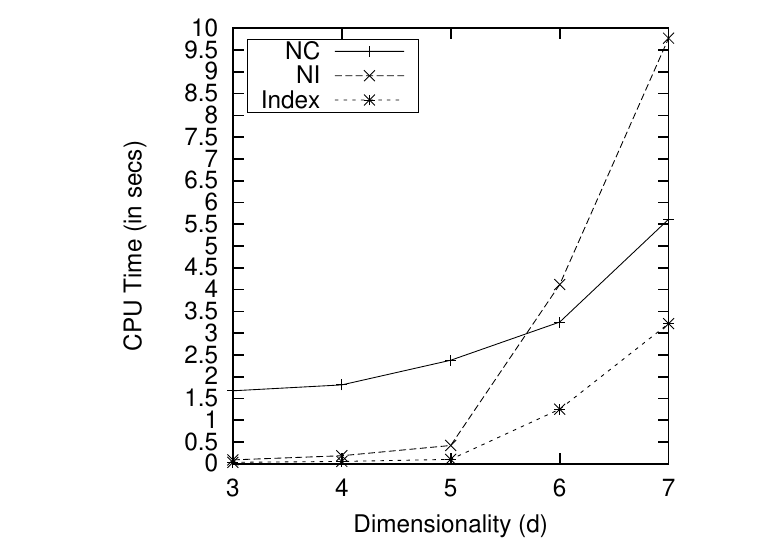}
&
	\includegraphics[width=\figwidth]{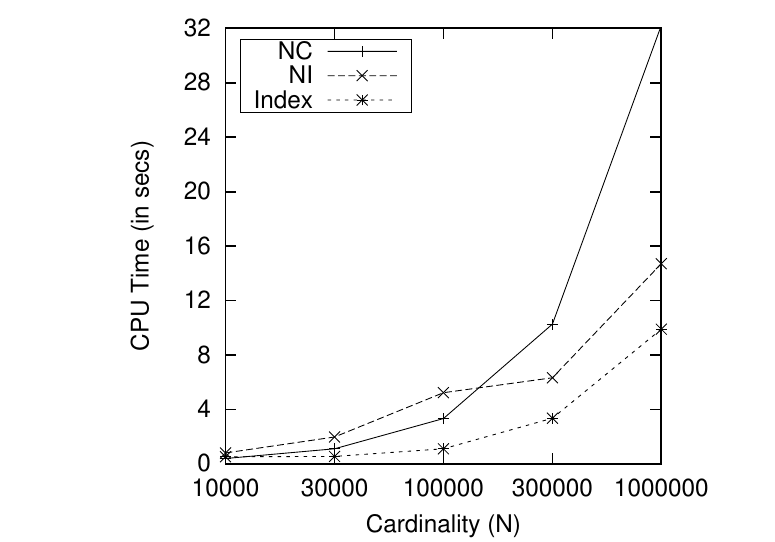}
\\
(a)&(b)
\end{tabular}
\caption{Effect of (a) dimensionality. (b) dataset cardinality.}
\label{fig:dim_cardi} 
\end{figure*}

\subsection{Synthetic datasets}

We used the standard data generator for skyline queries from
\url{http://www.pgfoundry.org/projects/randdataset} to generate synthetic
datasets; the dimensions were chosen to be independent.  
The scalability and performance of the techniques on synthetically generated data were
measured against four different parameters: (i)~cardinality of the dataset,
(ii)~dimensionality of the dataset, (iii)~size of the cache, and (iv)~number of
queries.  The values of these parameters were varied according to
Table~\ref{tab:expparams}.  To study the effect of one parameter, the other
parameters were held constant at the default values shown in bold.

Fig.~\ref{fig:dim_cardi}(a) shows the performance of the different techniques
with varying dimensionality.  As dimensionality increases, the cardinality of
the skyline set increases roughly exponentially for independent
datasets~\cite{1129924,8473992}.  The running time of the non-caching method
more or less shows the same behavior.  The number of semantic segments need to
be maintained increases exponentially as well.  Thus, when no index is used in
the cache, the running time is more than when index is used.  After $d = 5$, the
size of the cache is not enough to hold all the semantic segments, and many new
queries are classified as novel queries or partial queries.  Consequently, the
running time increases.

Fig.~\ref{fig:dim_cardi}(b) shows the effect of the cardinality of the dataset.  
For small datasets, the overhead of searching through all
the semantic segments makes the caching method slower than simply processing the
skylines from the database.  For larger datasets, the overhead becomes
negligible as compared to the gains of using the cache; consequently, the
non-caching technique performs the worst.  The indexing technique reduces this
search overhead and, hence, requires the least amount of time for all datasets.

\begin{figure*}[t]
\centering
\begin{tabular}{cc}
	\includegraphics[width=\figwidth]{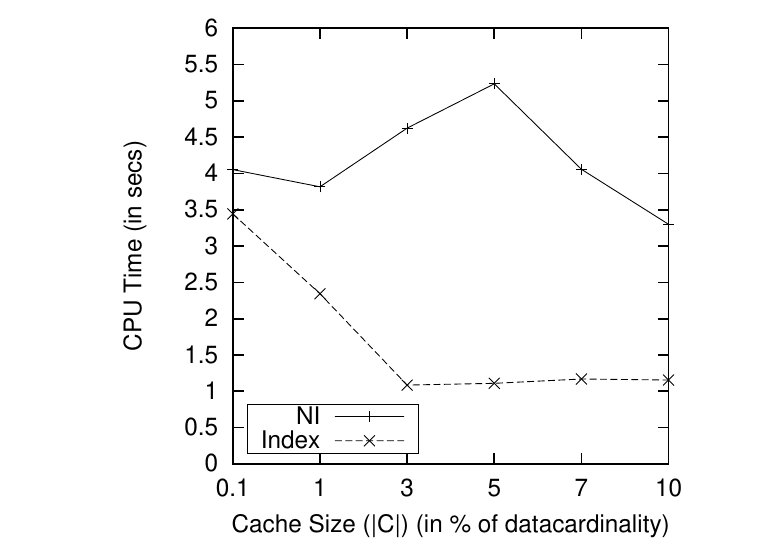}
&
	\includegraphics[width=\figwidth]{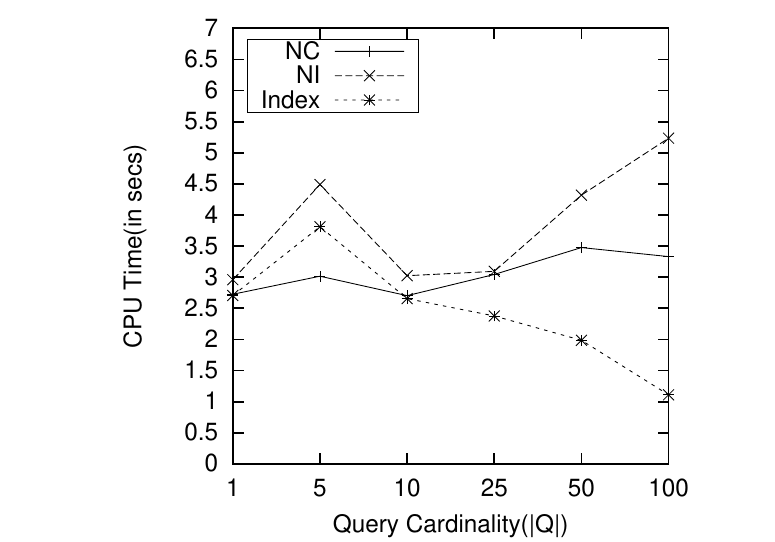}
\\
(a) & (b)
\end{tabular}
\caption{Effect of (a) cache size. (b) query cardinality.}
\label{fig:query_cache}
\end{figure*}

We next investigate the effect of cache size, measured as a percentage of the
size of the dataset.  Since the non-caching method does not depend on the size,
it is omitted from this experiment.  Fig.~\ref{fig:query_cache}(a) shows how the
running time is affected by varying cache size.  When the size is very small,
only a few semantic segments can be stored.  In such situations, indexing helps
only to a small extent.  As the size increases, indexing allows more semantic
segments to be stored because of the way a semantic segment shares its result
set with its subsets.  The non-indexing method, on the other hand, suffers from
processing too many semantic segments without much gain.  When the cache becomes
quite large, it allows most of the semantic segments to be stored along with
their result sets.  More queries can now be classified as exact or subset
queries and the performance of the non-indexing method improves.  The
performance of the indexing technique saturates and does not improve after a
point.

Ideally, when there is enough space in the cache, and the system has ``seen''
all possible skyline queries, any new query should be answered very fast.  The
final set of experiments tries to understand this phenomenon in more detail.  

Figure~\ref{fig:query_cache}(b) shows the average running time of a query as
more and more queries arrive.  When no caching is used, the number of queries do
not have any effect, and as expected, the average running time of a query varies
randomly.  For the first few number of queries, the cache is virtually empty,
and processing the cache yields no hits and no benefit at all.  In fact, the
overhead of maintaining the cache worsens the performance in comparison to the
no-caching technique.  Subsequently, as more queries arrive, the performance
improves for the indexing method.  However, when indexing is not used
and the semantic segments are left unorganized in the cache, lesser number of
semantic segments are stored due to redundancy of the result sets.  This leads
to less number of cache hits, and the performance suffers.

\begin{figure}[t]
\centering
\includegraphics[width=\figwidth]{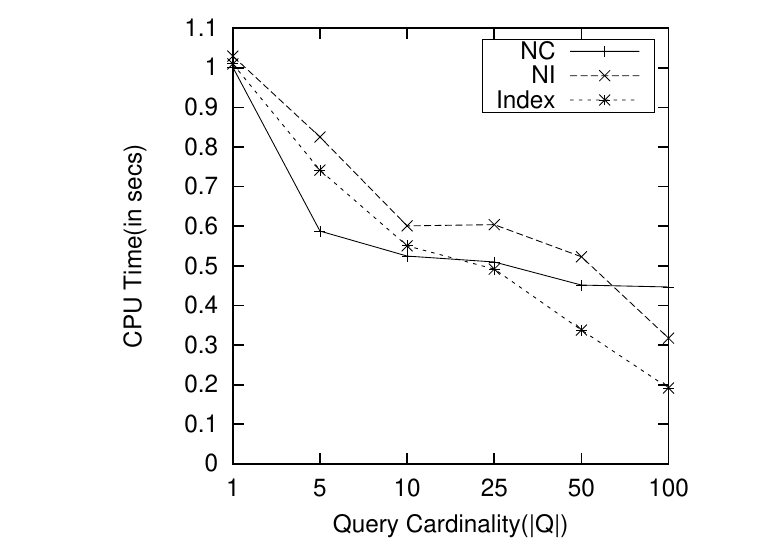}
\caption{Progressive performance of different techniques as more queries arrive.}
\label{fig:varq}
\end{figure}

\subsection{Real datasets}

We also tested the performance of the techniques on a real dataset from
\url{http://www.databasebasketball.com/}. The database provides the statistics
of NBA players with different attributes such as total points, assists, field
goals made, free throws made, etc.  Among these, six different dimensions were
chosen where the data is not missing for most of the players.  The cardinality
of the relation was 19,980.  The cache size was set to 5\% of that of the
relation.

The average running time of a query for the different techniques is plotted in
Fig.~\ref{fig:varq} against the number of queries.  While the time for the
non-caching technique stabilizes after a few queries, that for the caching
methods decreases.  Due to the superior organization of the semantic segments by
the indexing technique, the improvement is more pronounced as compared to the
non-indexing technique.

\section{Conclusions}
\label{sec:conclusion}

In this paper, we have introduced the concept of semantic caching to accelerate
a skyline query by classifying it as one of the four types---exact, subset,
partial and novel.  While the exact and subset queries are processed directly
from the cache, partial results for partial queries can be output from the cache
before resorting to the database for the full skyline set.  We also proposed an
index structure to effectively organize the past queries in the cache and
improve the efficiency of the methods.  Experimental results on synthetic and
real datasets showed the effectiveness and scalability of the methods.  In
future, we plan to handle update-intensive databases.

\pagebreak

{
\bibliographystyle{abbrv}
\balance
\bibliography{refer}
}

\end{document}